\documentclass[11pt,a4paper,oneside]{article}
\usepackage[top=3cm, bottom=3cm, left=2cm, right=2cm]{geometry}
\usepackage[english]{babel}
\usepackage[utf8x]{inputenc}
\usepackage[T1]{fontenc}

\usepackage{amsthm,amsmath,amssymb,mathrsfs}

\usepackage{bm}
\usepackage{graphicx}
\usepackage[colorinlistoftodos]{todonotes}
\usepackage[colorlinks=true, allcolors=blue]{hyperref}
\usepackage{booktabs,subcaption}
\usepackage{hyperref}
\usepackage[all]{nowidow}

\usepackage{setspace}
\usepackage[ruled,norelsize,vlined]{algorithm2e}
\usepackage{tikz} 

\usepackage[sf]{titlesec}
\usepackage{sectsty}
\allsectionsfont{\centering \normalfont\scshape} 

\usepackage[round]{natbib}

\usepackage{authblk}
\title{An enriched mixture model for functional clustering}
\author[1]{Tommaso Rigon}
\affil[1]{Department of Decision Sciences, Bocconi University, Milan}
\date{}
\newtheorem{theorem}{Theorem}
\newtheorem{corollary}{Corollary}

\theoremstyle{definition}

\begin{document}
\maketitle







\begin{abstract} There is an increasingly rich literature about Bayesian nonparametric models for clustering functional observations. However, most of the recent proposals rely on infinite-dimensional characterizations that might lead to overly complex cluster solutions. In addition, while prior knowledge about the functional shapes is typically available, its practical exploitation might be a difficult modeling task. Motivated by an application in e-commerce, we propose a novel enriched Dirichlet mixture model for functional data. Our proposal accommodates the incorporation of functional constraints while bounding the model complexity. To clarify the underlying partition mechanism, we characterize the prior process through a P\'olya urn scheme. These features lead to a very interpretable clustering method compared to available techniques. To overcome computational bottlenecks, we employ a variational Bayes approximation for tractable posterior inference. 
\end{abstract}

\section{Introduction}\label{sec:1}
A private company selling flight tickets is interested in understanding the preferences and the needs of its customers, to implement effective marketing strategies and to provide tailored solutions to its clients. In this specific industry, a major goal is to assess the interests of customers towards each flight route, which represents the functional unit in our analysis. The involved number of flight routes is quite large and therefore route-specific marketing actions are practically unfeasible, since they would require massive human interventions. A possible solution is to consider groups (clusters) of similar routes to allow the development of cluster-specific policies which have an impact on homogeneous segments of the market. Such a strategy is highly effective as long as the number of clusters is limited and the obtained groups have a clear interpretation. Indeed, an overly complex clustering solution would be of little practical interest in our setting, regardless the fact that it might constitute a better fit for the data. The enriched mixture model we propose is specifically designed to address this business requirement.

The entries of the dataset at our disposal are the number of times that each route has been searched on the company's website,  comprising a collection of weekly counts for each flight route. These longitudinal measurements are characterized by relevant temporal patterns that can be exploited to produce a finer partition of the market, compared to approaches based on static indicators. This is immediately evident from Figure~\ref{fig:1}, where the smoothed trajectories of two different routes are depicted. However, note that in our specific application we will work with standardized functional observations and not with the raw data of Figure~\ref{fig:1}. In fact,  we are interested in grouping functions with similar shapes and not in capturing their average levels. 

\begin{figure}[tbp]
\includegraphics[width=\textwidth]{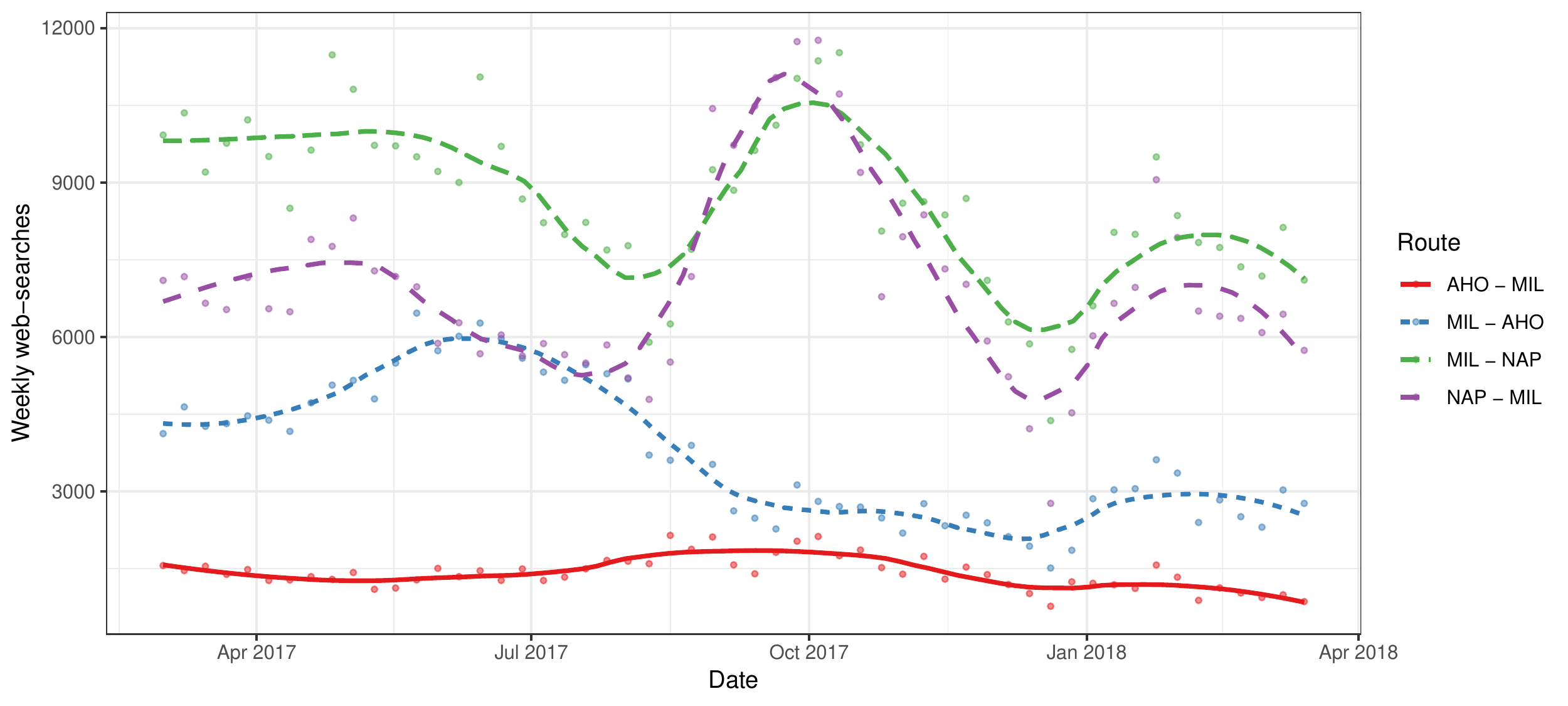}
\caption{Number of the web searches on an Italian website in the period between March 2017 and March 2018. The origin and the destination of each route are coded as follows: \texttt{MIL} = Milan, \texttt{NAP} = Naples, \texttt{AHO} =  Alghero. Smoothed trajectories are obtained using a nonparametrics \texttt{loess} estimate. \label{fig:1}}
\end{figure}
From a modeling perspective, we are given a collection of functional observations---one for each flight route---and we aim at partitioning them into groups. Direct application of classical procedures like \emph{k-means} or agglomerative methods seems inappropriate here. For example, they would disregard the temporal dimension and thus they would not take advantage from the functional structure of the data. Moreover, when the trajectories are observed on different time grids, or in presence of missing data, these tools cannot be employed. These considerations fostered the development of clustering procedures specifically designed for functional observations, see for instance \citet{Abraham2003, James2003, Serban2005} and references therein.

Let us assume that the route-specific measurements $y_i(t)$ can be regarded as error-prone realizations of unknown functions~$f_i(t)$, for each route $i=1,\dots,n$, and time value $t \in \mathbb{R}^+$, that is
\begin{equation}
\label{additive}
y_i(t)  = f_i(t) + \epsilon_i(t), \qquad i=1,\dots,n,
\end{equation}
with $\epsilon_i(t)$ denoting a random noise term, independent over flight routes and time. The additive specification~\eqref{additive} customarily serves as starting point in functional data analysis \citep{Ramsay2005}. Then, one could model the latent functions $f_i(t)$ separately using B-splines and subsequently grouping them using a \emph{k-means} algorithm on the regression coefficients \citep{Abraham2003}. Although such an approach is appealing because of its simplicity, it can not borrow strength across curves. Indeed, trajectories belonging to the same cluster are expected to behave similarly and therefore we should not discard this information from the analysis. In addition, with the \emph{k-means} approach one can not easily incorporate prior information on the functional shapes, which is indeed available in our setting. 

A natural way to fulfill the above requirements is through Bayesian mixtures. Functional clustering via finite mixtures have been provably effective in applications \citep[e.g.][]{Heard2006}, but question remains on the choice of mixture components, i.e. the number of clusters. 
A possible solution is to rely on Bayesian nonparametric priors, and one may follow \citet{Bigelow2009} who proposed a spline formulation for each $f_i$ together with the Dirichlet process prior of \citet{Ferguson1973} for the associated regression coefficients. Similarly, \citet{Ray2006} adopted the Dirichlet process in conjunction with wavelets. The resulting process is called functional Dirichlet process (\textsc{fdp}). In short, the \textsc{dp} prior induces a latent partition structure among the unknown functions $f_i$, while leaving unbounded the number of clusters, which increases logarithmically as $n$ grows. In \citet{Dunson2008} such a model has been employed for joint modeling of functional observations with a response variable, whereas in \citet{Petrone2009} a hybrid \textsc{fdp} is proposed, allowing realizations of $f_i(t)$ to share atoms in different local regions. Finally, refer also to \citet{Rodriguez2014} for the description of a functional generalized Dirichlet process model in nested designs. 

Although the latter methods enable flexible clustering and they are excellent tools for density estimation, their practical usage might be limited here. Indeed, the employment of a model with an unbounded number of groups might undermine the original goal, namely providing small dimensional summaries of flight routes.  Furthermore, all the above models seem to rely too much on data while ignoring accumulated knowledge from past analyses. For example, it is known that some flight routes are characterized by a strong cyclical component, e.g. the one depicted in Figure~\ref{fig:1}, and one may want to include this aspect in the model. The latter remark motivated \citet{Scarpa2009} to propose a contaminated \textsc{fdp} accounting for parametric functional specifications. Such an approach was then extended and theoretically investigated by \citet{Canale2017} in the more general Pitman--Yor case.

To overcome all the above limitations we propose an enriched functional Dirichlet multinomial process (\textsc{e-fdmp}), which has a \emph{bounded complexity} in terms of number of clusters and can easily incorporate prior knowledge about functional shapes.  
We will show that the proposed model converges to the enriched class of functional Dirichlet processes  (\textsc{e-fdp}) presented in \citet{Scarpa2014}, when the number of clusters is allowed to be infinite, while being also reminiscent of the enriched Dirichlet process of \citet{Wade2011}. Specifically, the underlying clustering mechanism can be described in terms of a two-step enriched urn-scheme, extending the well-know \citet{Blackwell1973} P\'olya urn. Such a theoretical development clarifies the interpretation of the involved random partition and it is helpful in the practical specification of the hyperparameters. 

The paper is organized as follows. Section~\ref{sec:2} introduces the enriched mixture model and Section~\ref{sec:3} discusses its enriched clustering mechanism. In Section~\ref{sec:4} a variational Bayes algorithm for posterior inference is developed and it is tested on a simulation study in Section~\ref{sec:5}. In Section~\ref{sec:6} we apply the proposed method to a real dataset from e-commerce.

\section{A Bayesian functional mixture model}\label{sec:2}

In the additive representation~\eqref{additive} we consider standardized functional observations. That is, the empirical mean of $y_i(t)$ evaluated on the time grid $\bm{t}_i = (t_{i1},\dots,t_{iT_i})^\intercal$ for $i=1,\dots,n$, equals zero, whereas the empirical variance equals one.  Then, for each standardized route and time value $t \in \mathbb{R}^+$, we let
\begin{equation*}
y_i(t)  = f_i(t) + \epsilon_i(t), \qquad i=1,\dots,n,
\end{equation*}
where each $f_i(t): \mathbb{R}^+ \rightarrow \mathbb{R}$ is an unknown function to be estimated, and where $\epsilon_i(t)$ is a Gaussian local error measurement with zero mean and variance~$\sigma^2$, in turns having a conditionally conjugate gamma prior distribution $\sigma^{-2} \sim \textsc{ga}(a_\sigma,b_\sigma)$. Consistent with the discussion of Section~\ref{sec:1}, we employ a discrete prior law $\tilde{p}$ to borrow information across  the latent trajectories $f_i(t)$ and to induce functional clustering, namely we assume
\begin{equation}\label{model}
(f_i(t) \mid \tilde{p}) \overset{\text{iid}}{\sim} \tilde{p}, \qquad \tilde{p} = \sum_{h=1}^H\xi_h\delta_{\phi_h(t)},
\end{equation}
independently for $i=1,\dots,n$, with $\delta_x$ denoting the point mass function at $x$. The collections of weights $\xi_1,\dots,\xi_H$ are random probabilities such that $\sum_{h=1}^H \xi_h=1$ almost surely, whereas each atom $\phi_h(t)$ is the realization of a random function. Hence, each $f_i(t)$ can be formally regarded as a random function belonging to a suitable complete and separable metric space $\mathbb{F}$ endowed with its Borel $\sigma$-algebra~$\mathscr{F}$. From representation \eqref{model} it is evident that a \emph{discrete} prior induces ties among the functions $f_i(t)$. We will say that two different functional observations $y_i(t)$ and $y_j(t)$ belong to same group whenever they possess the same functional atom $\phi_h(t)$, i.e. when they share the same latent trajectory $f_i(t) = f_j(t)$. Clearly, the choice of the prior law for $\tilde{p}$ has a strong impact on the clustering procedure. A popular class of models, arising in the infinite case $H \rightarrow \infty$, is given by stick-breaking priors \citep{Ishwaran2001}, of which the functional Dirichlet process (\textsc{fdp}) is a special case. However, as discussed in the Introduction and detailed in Section~\ref{sec:3}, such a choice might be unsuitable for our goals, and we rather want to upper-bound the model complexity by selecting a finite value for~$H$.  Furthermore, we aim at adapting~\eqref{model} to incorporate prior information about functional shapes. 

Suppose it is known that each $f_i(t)$ possesses specific shapes or features. For example, we may know in advance that a subset of the functional observations $f_i(t)$ is  monotone, cyclical or it is bounded by some constant. In our application, for instance, we know that a subset of routes presents a strong cyclical pattern. More formally, we assume that each function $f_i(t)$ belongs to a functional class among a finite collection $\{\mathbb{F}_1,\dots,\mathbb{F}_L\}$ of $L$ specifications, with each $\mathbb{F}_\ell \in \mathscr{F}$ being a measurable subset of $\mathbb{F}$. These functional classes have to be specified in consultation with subject matter experts or as a consequence of exploratory analyses. For example, one might want to consider either increasing, positive, periodical functions---or even \emph{biphasic} trajectories as in \citet{Scarpa2009}. Splines are particularly convenient in accommodating a variety of constraints such as monotonicity \citep{Ramsay1988}, but there are endless modeling possibilities. For instance, Gaussian processes are a flexible and widely used prior for functional modeling \citep[e.g.][]{Petrone2009}, and one may select for each class a different covariance function. A computationally convenient class of functions which includes the aforementioned examples is discussed in Section~\ref{sec:21}.

Let $P_\ell$ for $\ell=1,\dots,L$ be a collection of \emph{diffuse} probability measures defined over the space $(\mathbb{F}, \mathscr{F})$ and placing mass only on the corresponding class space $\mathbb{F}_\ell$, so that $P_{\ell}(\mathbb{F}_\ell) = 1$. The diffuseness assumption amount to have $P_\ell(\{f\}) = 0$ for any $f\in \mathbb{F}$. Then, our enriched formulation specializes the general model~\eqref{model} as follow
\begin{equation}\label{enriched_model}
\begin{aligned}
\tilde{p} &= \sum_{\ell=1}^L \Pi_\ell \sum_{h=1}^{H_\ell}\pi_{\ell h} \delta_{\theta_{\ell h}(t)}, \theta_{\ell h}(t) &\overset{\text{ind}}{\sim} P_\ell, \qquad h=1,\dots,H_\ell, \qquad \ell=1,\dots,L.
\end{aligned}
\end{equation}
Such a construction can be readily interpreted as a \emph{mixture of mixtures}. Differently from common mixture models, the atoms $\theta_{\ell h}(t)$ are independent and identically distributed (iid) within the feature class, but only independent across them. Exploiting standard hierarchical representation for mixture models, let us introduce a set of latent cluster indicators $\bm{G} = (G_1,\dots,G_n)$ whose values are the pairs $(\ell, h)$ for any $h=1,\dots,H_\ell$ and $\ell=1,\dots,L$, so that each function $f_i(t)$ is associated to the corresponding atom $\theta_{G_i}(t)$. Therefore, two functional observations $f_i(t)$ and $f_j(t)$ belong to the same cluster if and only if $G_i = G_j$. Moreover, let us define an additional set of latent indicators $F_i \in \{1,\dots,L\}$, for $i=1,\dots,n$, representing the membership of each $f_i(t)$ to the corresponding functional class. Then, the mixing probabilities in~\eqref{enriched_model} have a simple and useful interpretation, which is outlined in the following scheme:
\begin{equation*}
\begin{aligned}
\emph{Functional class allocation:}\quad &\mathbb{P}(F_i = \ell) = \Pi_\ell, \\
\emph{Within-class allocation:}\quad &\mathbb{P}(G_i = (\ell,h) \mid F_i = \ell) = \pi_{\ell h}, \ &&h=1,\dots,H_\ell,\\
\emph{Cluster allocation:}\quad &\mathbb{P}(G_i = (\ell,h)) = \Pi_\ell \pi_{\ell h}, \ &&h=1,\dots,H_\ell,\\
\end{aligned}
\end{equation*}
for any $\ell=1,\dots,L$ and unit $i=1,\dots, n$. To summarize, each membership indicator $G_i$ might be obtained as the result of a two-step procedure. In the first step, the functional class indicator $F_i$ associated to the $i$th unit is sampled according to the probabilities $\bm{\Pi}=(\Pi_1,\dots,\Pi_L)$. Then, conditionally on $F_i = \ell$, each cluster membership $G_i$ is drawn according to the within-class probabilities $\bm{\pi}_\ell = (\pi_{\ell 1},\dots,\pi_{\ell H_\ell})$. To allow uncertainty in such probabilities, we let
\begin{equation}\label{functional_class}
(\Pi_1,\dots,\Pi_{L-1}) \sim \textsc{dirichlet}(\alpha_1,\dots,\alpha_L),
\end{equation}
whereas for the within-class step we let
\begin{equation}\label{within_class}
(\pi_{\ell 1},\dots,\pi_{\ell H_\ell-1}) \overset{\text{ind}}{\sim} \textsc{dirichlet}\left(\frac{c_\ell}{H_\ell},\dots, \frac{c_\ell}{H_\ell}\right), \qquad \ell=1,\dots,L.
\end{equation}
The Dirichlet distribution in equation~\eqref{within_class} is symmetric because the atoms $\theta_{\ell h}$ are iid  within the functional class. Altogether, equations~\eqref{enriched_model}-\eqref{within_class} describe what we will term an enriched functional Dirichlet multinomial process (\textsc{e-fdmp}). 

Such a nested clustering mechanism characterizes general enriched priors, like the \textsc{e-fdp} and other enriched stick-breaking priors \citep{Scarpa2014}. As we will show in Section~\ref{sec:3}, there is a sharp connection between the \textsc{e-fdp} and our \textsc{e-fdmp}, since the former can be recovered as limiting case of the latter.  Beside constituting a more flexible class compared to classical mixtures, enriched processes allow the estimation of ``groups of clusters'', which are identified by the functional class indicators $F_i$.  Indeed, we might want to group the routes characterized by cyclical patterns or increasing trends, irrespectively of their within-class allocation. Moreover, even when the $G_i$ indicators are of interests, it might be useful to split the clustering solution into homogeneous classes, e.g. to facilitate their presentation to the stakeholders. These are major interpretative advantages of enriched priors which do not have a direct equivalent in classical mixture models.

\subsection{Baseline measures specification}\label{sec:21}

The specification of the baseline measures $P_\ell$ has clearly a crucial impact on inference. A priori, each $P_\ell$ can be interpreted as a ``functional prior guess'', because the expected value of $\tilde{p}$ is a mixture of the baseline measures $P_1,\dots,P_L$. Indeed, for any $A \in \mathscr{F}$
\begin{equation*}
\mathbb{E}\{\tilde{p}(A)\} = \sum_{\ell=1}^L \mathbb{E}(\Pi_\ell) P_\ell(A) = \frac{1}{\alpha} \sum_{\ell=1}^L \alpha_\ell P_\ell(A), \qquad \alpha = \sum_{\ell=1}^L\alpha_\ell.
\end{equation*}
The role of the hyperparameters $\alpha_1/\alpha,\dots,\alpha_L/\alpha$ is hence clear, being the prior proportions of each mixture component. For the remaining of the paper, we will focus on a broad subclass of baseline probability measures which are characterized by a significantly improved computational and analytical tractability. More precisely, we assume that $\theta_{\ell h}(t)$ is linear in the parameters, with a Gaussian prior on the regression coefficients, namely
\begin{equation}\label{linear_beta}
\theta_{\ell h}(t) = \sum_{m=1}^{M_\ell} \mathcal{B}_{m \ell}(t) \beta_{m\ell h}, \quad \bm{\beta}_{\ell h} = (\beta_{1\ell h},\dots,\beta_{M_\ell \ell h})^\intercal \overset{\text{ind}}{\sim} \mathcal{N}_{M_\ell}(\bm{\mu}_\ell, \bm{\Sigma}_\ell),
\end{equation}
where each $\mathcal{B}_{1\ell}(t),\dots,\mathcal{B}_{M_\ell \ell}(t)$ for $\ell=1,\dots,L$ is a set of pre-specified basis functions and where $\bm{\beta}_{\ell h} \in \mathbb{R}^{M_\ell}$ is an unknown vector of regression coefficients having multivariate Gaussian prior with mean $\bm{\mu}_\ell = (\mu_{1\ell},\dots,\mu_{M_\ell \ell})^\intercal$ and covariance matrix $\bm{\Sigma}_\ell$. Polynomials and splines might be used as basis functions, but the modeling possibilities are not confined to such a choice. For example, in our application we will employ trigonometric functions in combination with splines to capture perturbed cyclical patterns. Note that Bayesian penalized splines \citep{Lang2004} also fall within specification~\eqref{linear_beta}. Note that the a priori expected value of each function $f_i(t)$ for $i=1,\dots,n$ and $t\in \mathbb{R}^+$ simplifies, so that
\begin{equation*}
\mathbb{E}\{f_i(t)\} =   \sum_{\ell=1}^L \frac{\alpha_\ell}{\alpha} \sum_{m=1}^{M_\ell}  \mathcal{B}_{m\ell}(t) \mu_{m \ell},
\end{equation*} 
thus being a weighted average of the expected values of the regression coefficients. 
We shall remark that if inference on the functional classes $F_1,\dots,F_n$ is of interest, the measures $P_1,\dots,P_L$ must be distinguishable a priori, in the sense that they should characterize to quite different functional shapes. Otherwise, it might be difficult to infer the functional classes from the data. Indeed, while very flexible specifications might be employed for each $P_\ell$, these choices would lead to identifiability issues across functional classes. However, this is not a concern if one is interested in the cluster memberships $G_1,\dots,G_n$ and does not need to investigate also the class indicators $F_1,\dots,F_n$. 

\section{Random partitions and clustering}\label{sec:3}

In this section we investigate the a priori random partition mechanism of the \textsc{e-fdmp} model. Our proposal can be viewed as a middle ground between finite and infinite mixture models. Indeed, it is closely related to proper nonparametric priors while being finite dimensional. These features have several important implications for clustering. 

A key property of the \textsc{e-fdmp} model is that the number of clusters is bounded by $H = \sum_{\ell=1}^LH_\ell$. However, this does not imply that the actual number of clusters is equal to $H$, because some partitions might be empty. Indeed, to circumvent the issue of selecting the number of mixture components, one might consider a mixture model with a large $H$ and employ a sparse prior, thus effectively deleting the redundant mixture weights. Such an approach has been advocated by \citet{Malsiner2016}, on the ground of the asymptotic results of \citet{Rousseau2011}. The amount of shrinkage towards the upper bound $H$ or towards the single cluster solution is regulated by the sparse prior~\eqref{within_class}. Hence, the \textsc{e-fdmp} should not be regarded as a classical finite mixture model, because the number of clusters is inferred from the data and it should not be specified in advance. 

We begin our discussion by first pointing out relevant connections of our proposal with both the \textsc{e-fdp} and the \textsc{fdp} processes, and by providing some first intuitions about the role of each~$H_\ell$.  Consider the probability that two functions are assigned to the same cluster. More precisely, let $f_i$ and $f_j$ be two draws from a \textsc{e-fdmp} with $i\neq j$, then it is easy to check that  a priori
\begin{equation}\label{cocluster}
\mathbb{P}(f_i = f_j) = \sum_{\ell=1}^L \frac{\alpha_\ell(\alpha_\ell+1)}{\alpha(\alpha+1)}\frac{c_\ell + H_\ell}{c_\ell H_\ell + H_\ell}.
\end{equation}
The a priori probability of co-clustering of equation~\eqref{cocluster} is decreasing over $H_\ell$, i.e. the within-class upper bounds, and increasing over $c_\ell$, the within-class total mass parameter. Importantly, as each $H_\ell \rightarrow \infty$ for $\ell=1,\dots,L$, the probability of co-clustering converges to a strictly positive constant
\begin{equation*}
\lim_{H_\ell \rightarrow\infty}\mathbb{P}(f_i = f_j) = \sum_{\ell=1}^L \frac{\alpha_\ell(\alpha_\ell+1)}{\alpha(\alpha+1)}\frac{1}{1 + c_\ell},
\end{equation*}
which coincides with the co-clustering probability of the \textsc{e-fdp}, given in \citet{Scarpa2014}. Indeed, one can show that a \textsc{e-fdmp} (weakly) converges to a \textsc{e-fdp} as each $H_\ell\rightarrow \infty$. This convergence result has relevant practical implications: broadly speaking, it means that if we augment the model complexity indefinitely by increasing $H_\ell$, we nonetheless obtain a well-defined model, whose probability of co-clustering does not goes to zero. However, this is not to say that we should choose $H_\ell$ as large as possible, because this might lead to uninterpretable clustering solutions. Rather, the bounds $H_\ell$ should be selected as the largest value maintaining the model sufficiently tractable.  

We now provide a formal statement of the aforementioned convergence result, which rely on the notion of weak convergence for random measures; we refer to \citet[Chap. 4,][]{Kallenberg2017} for a rigorous treatment. Let $\tilde{q} \sim \textsc{dp}(c P)$ denote a Dirichlet process having total mass parameter $c$ and baseline probability distribution $P$ \citep{Ferguson1973}.

\begin{theorem}\label{teo1} Let $\tilde{p}$ be a \textsc{e-fdmp} defined by equations~\eqref{enriched_model}-\eqref{within_class} and let $\tilde{p}_\infty$ be a \textsc{e-fdp} \citep{Scarpa2014}, which is defined as
\begin{equation*}
\tilde{p}_\infty= \sum_{\ell=1}^L \Pi_\ell \tilde{q}_\ell,\qquad \tilde{q}_\ell \overset{\textup{ind}}{\sim} \textsc{dp}(c_\ell P_\ell),
\end{equation*}
where the probabilities $(\Pi_1,\dots,\Pi_L)$ are distributed as in~\eqref{functional_class}. Then, 
\begin{equation*}
\tilde{p} \overset{\text{w}}{\longrightarrow} \tilde{p}_\infty, \quad \text{as} \quad H_\ell \rightarrow \infty, \quad \ell =1,\dots,L,
\end{equation*}
where $\overset{\text{w}}{\longrightarrow}$ denotes weak convergence of the whole process.
\end{theorem}
\begin{proof} Note that we can write $\tilde{p}= \sum_{\ell = 1}^L \Pi_\ell \tilde{p}_{H_\ell}$, where each $\tilde{p}_{H_\ell}$ follows a Dirichlet multinomial process. It is well known that $\tilde{p}_{H_\ell}$ weakly converges to a Dirichlet process $q_\ell$ \citep[e.g.][]{Ishwaran2000} as $H_\ell\rightarrow \infty$, implying that for any finite collection of sets~$A_1,\dots,A_d \in \mathscr{F}$
\begin{equation*}
\{\tilde{p}(A_1),\dots,\tilde{p}(A_d)\} \overset{\text{d}}{\longrightarrow} \{\tilde{p}_\infty(A_1),\dots,\tilde{p}_\infty(A_d)\}.
\end{equation*}
Weak convergence of the process is a consequence of  Theorem~4.11 in \citet{Kallenberg2017}. 
\end{proof}

Theorem~\ref{teo1} is important also on the light of the following connection between the \textsc{e-fdp} and the \textsc{fdp} which, to the best of our knowledge, was not made explicit elsewhere. If $L=1$, then the \textsc{e-fdp} trivially  reduces to a \textsc{fdp}. However, this occurs also under specific hyperparameter settings.  Indeed, the next corollary implies that if $\alpha_\ell = c_\ell$ for $\ell=1,\dots,L$, then the limiting process $\tilde{p}_\infty$ will be distributed according to a Dirichlet process whose baseline probability measure is a mixture of the class-specific measures $P_1,\dots,P_L$. Such a result is stated as a corollary of Theorem~\ref{teo1} for the sake of the exposition, but it is actually a property of the \textsc{e-fdp}; see the proof for details.

\begin{corollary}\label{corol1} Suppose additionally to Theorem~\ref{teo1} that $\alpha_\ell = c_\ell$ for any $\ell=1,\dots,L$. Then $\tilde{p} \overset{\text{w}}{\longrightarrow} \tilde{p}_\infty$ as each $H_\ell \rightarrow \infty$ and moreover
\begin{equation*}
\tilde{p}_\infty \sim \textsc{dp}\left(\sum_{\ell=1}^L \alpha_\ell P_\ell\right).
\end{equation*}
\end{corollary}
\begin{proof} The proof rely on the finite-dimensional characterization of the Dirichlet process \citep{Ferguson1973}. Specifically,  for any finite partition $B_1,\dots,B_d \in \mathscr{F}$ we have 
\begin{equation*}
\{\tilde{q}_\ell(B_1),\dots,\tilde{q}_\ell(B_d)\} \overset{\textup{ind}}{\sim} \textsc{dirichlet}\{\alpha_\ell P_\ell(B_1),\dots, \alpha_\ell P_\ell(B_d)\},\quad \ell=1,\dots,L.
\end{equation*} 
Note that $\{\tilde{p}_\infty(B_1),\dots,\tilde{p}_\infty(B_d)\} = \sum_{\ell=1}^L \Pi_\ell \{\tilde{q}_\ell(B_1),\dots,\tilde{q}_\ell(B_d)\}$, and
\begin{equation*}
\{\tilde{p}_\infty(B_1),\dots,\tilde{p}_\infty(B_d)\} \sim \textsc{dirichlet}\left\{ \sum_{\ell=1}^L \alpha_\ell P_\ell(B_1),\dots,  \sum_{\ell=1}^L \alpha_\ell P_\ell(B_d)\right\},
\end{equation*} 
thanks to well-know properties of the Dirichlet distribution. 
\end{proof}

\subsection{Enriched P\'olya urn scheme}\label{sec:31}

Similar to \citet{Blackwell1973} in the Dirichlet process case, our \textsc{e-fdmp} is characterized by a P\'olya urn scheme, whose description greatly facilitates the understanding of the underlying clustering mechanism. Conditionally on the latent class indicators $F_1,\dots,F_n$, our enriched formulation reduces to a collection of Dirichlet multinomial processes. Recalling equation~\eqref{enriched_model}, we can rewrite the \textsc{e-fdmp} as follows
\begin{equation*}
\tilde{p}= \sum_{\ell = 1}^L \Pi_\ell \tilde{p}_{H_\ell}, \qquad \tilde{p}_{H_\ell} = \sum_{h=1}^{H_\ell} \pi_{\ell h}\delta_{\theta_{\ell h}(t)}.
\end{equation*}
Then, we can augment the above specification by including the set of latent class indicators $\bm{F} = (F_1,\dots,F_n)$. In this hierarchical representation, the functions  belonging the same class $f_i : i \in \mathcal{I}_\ell$ with $\mathcal{I}_\ell = \{i =1,\dots,n : F_i = \ell\}$ are iid draws from $\tilde{p}_{H_\ell}$, a Dirichlet multinomial process. More precisely, we can equivalently represent our \textsc{e-fdmp} hierarchically as
\begin{equation*}\label{conditionalF}
\begin{aligned}
(F_i \mid \bm{\Pi} ) &\overset{\text{iid}}{\sim} \textsc{multinom}(\Pi_1,\dots,\Pi_L), \qquad &&i=1,\dots,n, \\
(f_i \mid F_i = \ell, \tilde{p}_{H_\ell}) &\overset{\text{iid}}{\sim} \tilde{p}_{H_\ell}, \qquad &&i \in \mathcal{I}_\ell
\end{aligned}
\end{equation*}
 with prior distributions as in equations~\eqref{functional_class}-\eqref{within_class}. Such a hierarchical representation naturally leads to the definition of a sequential mechanism for generating both $f_1,\dots,f_n$ and $F_1,\dots,F_n$. Let $n_\ell = \sum_{i=1}^n I(F_i = \ell)$ be the number of elements belonging to the $\ell$th functional class and let $k_\ell \le n_\ell$ be the number of distinct values observed among the functions of the $\ell$th class. Moreover, let $f^*_{11},\dots,f^*_{1n_1},\dots,f^*_{L1},\dots,f^*_{Ln_L}$ represent the distinct values observed in the whole sample $\bm{f} = (f_1,\dots,f_n)$, having frequencies $n_{j\ell}$ for $j=1,\dots,k_\ell$ and $\ell=1,\dots,L$, so that $n_\ell = \sum_{j=1}^{k_\ell} n_{j\ell}$ and $n = \sum_{\ell=1}^L n_\ell$. Then, the enriched P\'olya urn scheme is characterized by the following two steps, so that for any $n\ge 1$ and any $A \in \mathscr{F}$ we have 
\begin{equation*}\label{enriched_polya}
\begin{aligned}
\mathbb{P}(F_{n+1} = \ell \mid \bm{F}) &= \frac{\alpha_\ell + n_\ell}{\alpha + n}, \qquad \ell=1,\dots,L,\\
\mathbb{P}(f_{n+1} \in A \mid \bm{f}, \bm{F}, F_{n+1} = \ell) &= \left(1 - \frac{k_\ell}{H_\ell}\right) \frac{c_\ell}{c_\ell + n_\ell}P_\ell(A) + \sum_{j=1}^{k_\ell}\frac{n_{j\ell} + c_\ell/H_\ell}{c_\ell + n_\ell} \delta_{f^*_{j\ell}}(A).
\end{aligned}
\end{equation*}
At the first step, one draws the $F_{n+1}$ functional class indicator with a probability depending on the observed frequencies $n_1,\dots,n_L$ and the $\alpha_1,\dots,\alpha_L$ coefficients, which can be  naturally interpreted as a priori frequencies. Then, at the second step and given $F_{n+1} = \ell$, one either draw a novel functional observation from $P_\ell$ or she samples one of the previously observed functions with probability proportional to $n_{j\ell} + c_\ell/H_\ell$. On the light of Theorem~\ref{teo1}, it is not surprising that the second step converges to the classical scheme of \citet{Blackwell1973} as $H_\ell \rightarrow \infty$, conditionally on the $\ell$th functional class. Moreover, if $\alpha_\ell = c_\ell$ the classical P\'olya urn scheme is recovered also marginally, a consequence of Corollary~\ref{corol1}. Furthermore, such an enriched P\'olya urn scheme is reminiscent of the one presented in \citet{Wade2011}, and indeed it can be essentially regarded as its finite-dimensional counterpart. 

Let us focus on the conditional probability of obtaining a new cluster, given the functions $\bm{f}$ and the class indicators $\bm{F}$. From the enriched P\'olya urn scheme one can easily get
\begin{equation}\label{predictive}
\mathbb{P}(f_{n+1} = \text{``new"} \mid \bm{f}, \bm{F}) = \sum_{\ell=1}^L \frac{\alpha_\ell + n_\ell}{\alpha + n}\ \left(1 - \frac{k_\ell}{H_\ell}\right) \frac{c_\ell}{c_\ell + n_\ell}.
\end{equation}
The above predictive probability provides a clear guidance about the role of the hyperparameters. In first place, note that the probability of drawing a new function decreases the more clusters $k_\ell$ we observe, and it equals zero whenever $k_\ell = H_\ell$. Hence, the \textsc{e-fdmp} penalizes partitions with a large number of clusters, effectively bounding the model complexity, one of the overarching goals of our analysis. Note that as $H_\ell \rightarrow \infty$ the aforementioned penalization disappears. Moreover, the parameters $c_\ell$ control the creation of a new cluster---the larger each $c_\ell$ the more cluster we should expect. 

\section{Posterior computations}\label{sec:4}

Bayesian mixture models are routinely estimated using Markov chain Monte Carlo (\textsc{mcmc}). While this approach is supported by strong theoretical guarantees, it has some drawbacks when performing clustering. The first concern is scalability: \textsc{mcmc} sampling might face computational bottlenecks when the sample size grows. This is a severe limitation because in practice one would like to conduct the clustering algorithm on a weekly basis, and perhaps on several different datasets. In addition, a further difficulty arises when performing clustering with \textsc{mcmc}. As discussed  in \citet{LaGr2007}, at each step of the chain one samples a different partition of the observations; however, it is hard to provide a point estimate, essentially because of the label switching phenomenon. Existing solutions rely either on ad-hoc procedures \citep{Medvedovic2002}, or on post-process optimizations problems \citep{LaGr2007, Fric2009, Wade2018}. In both cases, this implies an additional layer of difficulty that one might want to avoid.

To address these issues we employ a mean-field variational approximation of the posterior distribution, which is nowadays a standard choice in several fields \citep{Blei2017}. The involved computations are much faster than \textsc{mcmc}, and the variational Bayes (\textsc{vb}) approach is particularly well suited for clustering purposes, since it is not affected by label switching, thus ruling out the aforementioned additional steps. In addition, variational inference for the \textsc{e-fdmp} is straightforward to implement because such a model belongs to the conditionally conjugate exponential family, for which efficient optimization algorithms are available \citep{Blei2017}. Unfortunately, these advantages do not come without some drawbacks: indeed, the variational posterior is often a crude approximation of the proper posterior law, and it is well known that \textsc{vb} generally leads to accurate point estimates but also it typically underestimates the variability. If uncertainty quantification were of interest, a Gibbs sampling algorithm for the \textsc{e-fdmp} could be easily devised, since the full conditional distributions are be available in closed form. However, in our motivating application we are only interested in a single cluster solution and therefore \textsc{vb} represents an appealing choice. 

Let $\bm{\pi} = (\bm{\pi}_1,\dots, \bm{\pi}_L)$ be the collection of the within-class probabilities of equation~\eqref{within_class} and let $\bm{\beta} = (\bm{\beta}_{11},\dots,\bm{\beta}_{1 H_1}, \dots, \bm{\beta}_{L 1}, \dots, \bm{\beta}_{L H_L})$ be the set of regression coefficients appearing in equation \eqref{linear_beta}. We seek a variational distribution $q(\bm{G},\bm{\Pi},\bm{\pi}, \bm{\beta},\sigma^2)$ that best approximates the joint posterior, while maintaining simple computations. This can be obtained by minimizing the Kullback-Leibler divergence between the variational distribution and the full posterior, or equivalently by maximizing the so-called evidence lower bound (\textsc{elbo}); see \citet{Blei2017}. Without further restrictions, the Kullback-Leibler divergence is minimized when the variational distribution is equal to the true posterior distribution, which is analytically intractable. Hence, a common strategy  is to assume that the variational distribution belongs to a mean-field family. Such a class of distributions incorporate a posteriori independence among distinct groups of parameters, meaning that the variational distribution factorizes as
\begin{equation*}
q(\bm{G},\bm{\Pi},\bm{\pi}, \bm{\beta},\sigma^2) = q(\sigma^2)\prod_{i=1}^n q(G_i) q(\bm{\Pi}) \prod_{\ell=1}^L q(\bm{\pi}_\ell) \prod_{\ell=1}^L\prod_{h=1}^{H_\ell}q(\bm{\beta}_{\ell h}) .
\end{equation*}
Under such an assumption, the optimal variational distributions can be found exploiting an iterative algorithm called coordinate ascent variational inference (\textsc{cavi}). Its full derivation entails standard calculations  which are omitted for the sake of the exposition; we report in Algorithm~\ref{algo1} only the resulting \textsc{cavi} algorithm. One may refer to \citet[Chap. 10,][]{Bishop2006} for detailed illustrations on similar models.  

We define here some additional notation necessary for the description of the \textsc{cavi}~Algorithm~\ref{algo1}. As mentioned in Section~\ref{sec:2}, recall that each functional observation $y_i(t)$ is only available on a finite grid of points $\bm{t}_i = (t_{i1},\dots,t_{iT_i})^\intercal$. The observed values associated to these time grids are stacked into a single $\sum_{i=1}^n T_i$-dimensional vector 
\begin{equation*}
\bm{y} = (y_1(t_{11}),\dots, y_1(t_{1T_1}),\dots, y_n(t_{n1}),\dots, y_n(t_{n T_n}))^\intercal.
\end{equation*}
Similarly, we define the $\sum_{i=1}^n T_i \times M_\ell$ matrices $\bm{B}_\ell$ for $\ell=1,\dots,L$, which are paired to the data $\bm{y}$ and whose entries are the values of the basis functions $\mathcal{B}_m(t_{is})$ of equation~\eqref{linear_beta}, for $m=1,\dots,M_\ell$ over the columns and for $s=1,\dots,T_i$ and $i=1,\dots,n$ over the rows. Moreover, note that in Algorithm~\ref{algo1} the density functions are identified by the same symbols that are used to characterize distributions. Finally, the expected values appearing in Algorithm~\ref{algo1} are taken with respect to the variational distributions $q(\cdot)$ at the $r$th step of the cycle, motivating the notation $\mathbb{E}_q$.

From the output of the \textsc{cavi} algorithm, it is straightforward to derive a posteriori variational estimates for the cluster memberships $G_1,\dots,G_n$, for the class-specific membership $F_1,\dots,F_n$, and for the cluster-specific trajectories $\theta_{\ell h}(t)$.  A natural variational Bayes estimate $\hat{G}_1,\dots,\hat{G}_n$ for the cluster memberships is given by
\begin{equation*}
\hat{G}_i = \arg \max_{\ell, h} \rho_{i \ell h} =  \arg \max_{\ell, h} q\{G_i = (\ell, h)\},\qquad i=1,\dots,n,
\end{equation*}
 and similarly a variational estimate $\hat{F}_1,\dots,\hat{F}_n$ for the functional classes is
 \begin{equation*}
\hat{F}_i = \arg \max_{\ell} \sum_{h=1}^{H_\ell} \rho_{i \ell h} =  \arg \max_{\ell} q(F_i = \ell), \qquad i=1,\dots,n.
 \end{equation*}
These natural estimators can not be easily computed when performing \textsc{mcmc} because of the label-switching phenomenon. Finally, an estimate $\hat{\theta}_{\ell h}(t)$ for the cluster-specific functions is given by its variational expectation, which equals
\begin{equation*}
\hat{\theta}_{\ell h}(t) = \mathbb{E}_q\{\theta_{\ell h}(t)\} = \sum_{m=1}^{M_\ell} \mathcal{B}_{m \ell}(t) \mathbb{E}_q(\beta_{m \ell h}) = \sum_{m=1}^{M_\ell} \mathcal{B}_{m \ell}(t) \tilde{\mu}_{m \ell h},
\end{equation*} 
 where the vector of means $\tilde{\bm{\mu}}_{\ell h} = (\tilde{\mu}_{1\ell h},\dots, \tilde{\mu}_{M_\ell \ell h})^\intercal$ is the same obtained at Step~4 of Algorithm~\ref{algo1}. The estimate $\hat{\theta}_{\ell h}(t)$ could be useful for the interpretation of the clusters as well as for model checking. 
 
 \begin{algorithm*}[tp]
\caption{\textsc{cavi} algorithm for the \textsc{e-fdmp} \label{algo1}} 

\small{
\Begin{

\vspace{3pt}
Let $q(\cdot)$ denote the generic variational distribution at iteration $r$ and let $\mathbb{E}_q$ denote the expected value taken with respect to it. At every step of the algorithm, update each block of $q(\cdot)$ according to the following steps:

\vspace{3pt} 
{\bf [1]} Update $q(G_i)$ for each $i=1, \ldots, n$;

\vspace{3pt}
\For(){$i$ from $1$ to $n$}{

\vspace{3pt}
Update the variational probabilities $q\{G_i = (\ell, h)\} = \rho_{i \ell h}$ according to
\begin{equation*}
\begin{aligned}
\rho_{i \ell h} &\propto \exp{\left[\mathbb{E}_q\{\log{(\Pi_\ell \pi_{\ell h})}\} + \sum_{s=1}^{T_i}\mathbb{E}_q\{\log{\mathcal{N}(y_i(t_{is}); \theta_{\ell h}(t_{is}), \sigma^2)}\} \right]}, \\
&\propto \exp{\left(\mathbb{E}_q\{\log{(\Pi_\ell \pi_{\ell h})}\} - \frac{1}{2}\mathbb{E}_q(\sigma^{-2}) \sum_{s=1}^{T_i}\mathbb{E}_q\left[\left\{y_i(t_{is}) - \theta_{\ell h}(t_{is})\right\}^2\right]\right)},
\end{aligned}
\end{equation*}
for any $h=1,\dots,H_\ell$ and $\ell=1,\dots,L$.
}

\vspace{3pt} 
{\bf [2]} Update the  variational distribution $q(\bm{\Pi})$ according to
\begin{equation*}
q(\bm{\Pi}) = \textsc{dirichlet}\left(\bm{\Pi}; \alpha_1 + \sum_{i=1}^n \sum_{h=1}^{H_1} \rho_{i 1 h},\dots, \alpha_L + \sum_{i=1}^n \sum_{h=1}^{H_L} \rho_{i L h}\right).
\end{equation*}

\vspace{3pt}
{\bf [3]} Update $q(\bm{\pi}_\ell)$ for each $\ell=1, \ldots, L$;

\vspace{3pt}
\For(){$\ell$ \mbox{from} $1$ to $L$}{

\vspace{3pt}
Update the variational distribution of each $q(\bm{\pi}_\ell)$ according to
\begin{equation*}
q(\bm{\pi}_\ell) = \textsc{dirichlet}\left(\bm{\pi}_\ell; \frac{c_\ell}{H_\ell} + \sum_{i=1}^n \rho_{i \ell 1},\dots, \frac{c_\ell}{H_\ell} + \sum_{i=1}^n \rho_{i \ell H_{\ell}}\right).
\end{equation*}
}

\vspace{3pt}
{\bf [4]} Update $q(\bm{\beta}_{\ell h})$ for each $h=1,\dots, H_\ell$ and $\ell=1, \ldots, L$;

\vspace{3pt}
\For(){$\ell$ \mbox{from} $1$ to $L$}{

\vspace{3pt}
	\For(){$h$ \mbox{from} $1$ to $H_\ell$}{
	
	\vspace{3pt}
	Update the variational distribution of each $q(\bm{\beta}_{\ell h})$ according to
	
	\begin{equation*}
	q(\bm{\beta}_{\ell h}) = \mathcal{N}_{M_\ell}\left(\bm{\beta}_{\ell h}; \tilde{\bm{\mu}}_{\ell h}, \tilde{\bm{\Sigma}}_{\ell h}\right),
	\end{equation*}
	where $\tilde{\bm{\Sigma}}_{\ell h} = (\bm{B}_\ell^\intercal \bm{\Gamma}_{\ell h} \bm{B}_\ell + \bm{\Sigma}_\ell^{-1})^{-1}$ and  $\tilde{\bm{\mu}}_{\ell h} = \tilde{\bm{\Sigma}}_{\ell h}(\bm{B}_\ell \bm{\Gamma}_{\ell h} \bm{y} + \bm{\Sigma}_\ell \bm{\mu}_\ell)$, and with $\bm{\Gamma}_{\ell h} = \mathbb{E}_q(\sigma^{-2})\text{diag}(\rho_{1 \ell h},\dots,\rho_{1 \ell h},\dots,\rho_{n \ell h},\dots, \rho_{n \ell h})$.
	}

}
\vspace{3pt} 
{\bf [5]} Update the  variational distribution $q(\sigma^{-2})$ according to
\begin{equation*}
q(\sigma^{-2}) = \textsc{ga}\left(\sigma^{-2}; a_\sigma  + \frac{1}{2}\sum_{i=1}^n T_i,  b_\sigma + \frac{1}{2} \sum_{i=1}^n \sum_{s=1}^{T_i} \sum_{\ell =1}^L \sum_{h=1}^{H_\ell} \rho_{i \ell h} \mathbb{E}_q[\{y_i(t_{is}) - \theta_{\ell h}(t_{is})\}^2] \right).
\end{equation*}
}
}
\end{algorithm*}
 
\section{Simulated illustration}\label{sec:5}

In this section we assess the empirical performance of the \textsc{e-fdmp}---and the associated \textsc{cavi} algorithm---by conducting a simple simulation study. Such a simulation is far from being extensive and it serves mainly as an illustration of the concepts presented in Section~\ref{sec:3}. Specifically, we aim at showing the ability of our model to effectively recover the true number of groups, as well as the cluster memberships, thereby empirically validating the role of each parameter $H_\ell$ as the upper bound for the total number of clusters. 

For this illustrative example, we consider identical and equally spaced time grids $\bm{t}_i = (1/T_i,\dots,T_i/T_i)^\intercal$ for $i=1,\dots,n$, ranging over the unit interval $[0,1]$, and we let the number of observations $n = 100$ and each grid length $T_1 = \cdots = T_n =  50$. Among the functions $f_1,\dots,f_n$ there are only four distinct values $f_1^*,\dots,f^*_4$, defined as
\begin{equation*}
\begin{aligned}
&f^*_1(t) = 1 - 2t, \qquad &&f^*_2(t) = \frac{1}{2}\{\cos(2 \pi t) + \sin(2 \pi t)\}, \\
&f^*_3(t) = 2t^4 -1, \qquad &&f^*_4(t) = \frac{1}{2}\{\cos(4 \pi t) + \sin(4 \pi t)\}.
\end{aligned}
\end{equation*}
The first $f_1,\dots, f_{25}$ functions are set equal to $f^*_1$, while each element of the second block $f_{26},\dots,f_{50}$ is set equal $f^*_2$, and similarly for the third and fourth blocks of functions $f_{51},\dots, f_{75}$ and $f_{76},\dots,f_{100}$, whose elements are equal to $f^*_3$ and $f^*_4$, respectively.  Summarizing, we let the number of cluster be equal to $4$ and we assume that each partition has $25$ elements, for a total of $n=100$ functional observations. Recall that we observe error prone realizations $y_i(t)$ of these functions under Gaussian noise, for $i=1,\dots,n$, as for equation~\eqref{additive}. Clearly, the clustering performance is affected by the amount of noise in the observed  data. To emphasize this aspect we consider two different scenarios. In the first simulated setting, the error variance is relatively small ($\sigma^2 = 0.1^2)$, while in the second scenario the functions are perturbed by a much higher amount ($\sigma = 1.5^2$). The simulated trajectories are depicted in Figure~\ref{fig:sim}: in the first scenario the four functions $f^*_1,\dots,f^*_4$ are clearly distinguishable, whereas in the latter the underlying signal is less evident. Consequently, the clustering algorithm is expected to perform better in the small variance setting than in the high variance one.

\begin{figure}[tp]
\centering
\includegraphics[width=\textwidth]{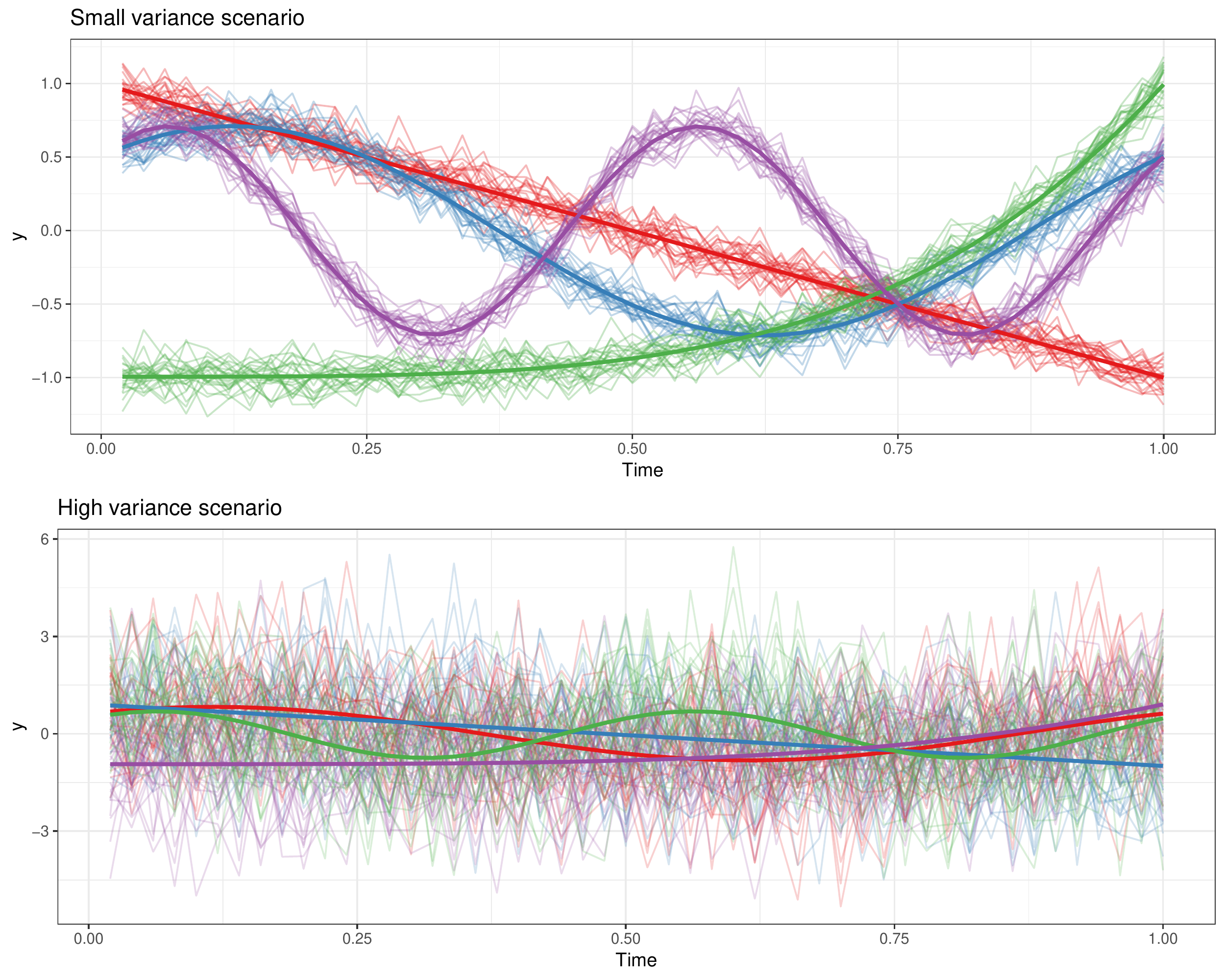}
\caption{Simulated trajectories $y_1(t),\dots,y_n(t)$ in the small variance scenario (top graph, $\sigma^2 = 0.1^2)$, and high variance scenario (bottom graph, $\sigma^2 = 1.5^2$). Different colors refer to the estimated cluster memberships $\hat{G}_1,\dots,\hat{G}_n$ whereas the corresponding solid lines are the estimated cluster-specific functions $\hat{\theta}_{\ell h}(t)$. 
\label{fig:sim}}
\end{figure}

Although the true number of clusters is $4$, we set the total number of mixture components $H = 20$, to empirically demonstrate the ability of the \textsc{e-fdmp} to recover the correct number of distinct functions. Moreover, we let the number of class functions $L = 4$ and each within-class upper bound $H_\ell = 5$ for $\ell=1,\dots,4$. The functional atom specifications and the corresponding basis functions $\mathcal{B}_{m\ell}(t)$, as for equation~\eqref{linear_beta}, are the following
\begin{equation*}
\begin{aligned}
&\theta_{1h}(t) = \beta_{11h} + \beta_{21h}t, \ && \theta_{2h}(t) = \beta_{12h} + \beta_{22h}\cos(2 \pi t) + \beta_{32h}\sin(2 \pi t), \\
&\theta_{3h}(t) = \beta_{13h} + \beta_{23h}t^4, \ &&\theta_{4h}(t) = \beta_{14h} + \beta_{24h}\cos(4 \pi t) + \beta_{34h}\sin(4 \pi t),
\end{aligned}
\end{equation*}
with iid prior distributions $\beta_{m\ell h} \overset{\textup{iid}}{\sim} \mathcal{N}(0,10)$. The prior specification is concluded by setting $\alpha_1 = \cdots = \alpha_L = 1$, $c_1=\cdots=c_L = 1$ and $a_\sigma = b_\sigma = 1$.

The optimization of the \textsc{elbo} might be troublesome due to the presence of local maxima. To mitigate this issue, the \textsc{cavi} algorithm was initialized at several different starting points; the solution achieving the highest value of the \textsc{elbo} was retained \citep{Blei2017}. Remarkably, each run of the \textsc{cavi} required only few seconds for the computations on a standard laptop and with a na\"ive implementation in the R statistical software. The results are depicted in Figure~\ref{fig:sim} for both the scenarios. 

\begin{table}[t]
\begin{subtable}[h]{0.48\textwidth}
\centering
\begin{tabular}{lrrrr}
 \toprule
Class label &  $1$ &  $2$ &  $3$ &  $4$ \\ 
\midrule
Within-class label&  $1$ &  $2$ &  $3$ &  $3$ \\ 
  \midrule
$f^*_1$ &  25 &   0 &   0 &   0 \\ 
$f^*_2$ &   0 &  25 &   0 &   0 \\ 
$f^*_3$ &   0 &   0 &  25 &   0 \\ 
$f^*_4$ &   0 &   0 &   0 &  25 \\ 
   \bottomrule
\end{tabular}
   \caption{Small variance scenario. \label{tab:cluster1}}
\end{subtable}
\hfill
\begin{subtable}[h]{0.48\textwidth}
\centering
\begin{tabular}{lrrrr}
  \toprule
Class label& $1$ &  $2$ &  $3$ & $4$ \\ 
\midrule
Within-class label& $4$ & $2$ &  $2$ &  $4$ \\ 
  \midrule
 $f^*_1$ &  22 &   1 &   0 &   2 \\ 
 $f^*_2$ &  3 &  19 &   1 &   2 \\ 
 $f^*_3$ &   0 &   2 &  23 &   0 \\ 
 $f^*_4$ &  1 &   0 &   0 &  24 \\ 
   \bottomrule
\end{tabular}
  \caption{High variance scenario. \label{tab:cluster2}}
  \end{subtable}
    \caption{Contingency tables showing the true cluster memberships $G_1,\dots,G_n$ against the estimated memberships $\hat{G}_1,\dots,\hat{G}_n$ in the small variance (a) and in the high variance (b) scenarios. The functional class  and the within-class labels are reported. The cluster labels having zero frequencies are omitted.   \label{tab:cluster}}
\end{table}

In the small variance setting (top graph of Figure~\ref{fig:sim}), the \textsc{cavi} algorithm applied to the \textsc{e-fdmp} model performs remarkably well. Indeed, it correctly identifies $4$ clusters---meaning that  among the estimated memberships $\hat{G}_1,\dots,\hat{G}_n$ there are only $4$ distinct values---even though a conservative upper bound $H = 20$ was selected. Moreover, the observed curves are always allocated to the correct cluster, as summarized in Table~\ref{tab:cluster1}, up to a label permutation. Finally, the estimated curves $\hat{\theta}_{\ell h}$ depicted in Figure~\ref{fig:sim} closely resemble the true functions $f^*_1,\dots,f^*_4$.  Similar remarks can be made also in the high variance scenario (bottom graph of Figure~\ref{fig:sim}), although the performance are less striking, as one would expect. In particular, according to Table~\ref{tab:cluster2}  the estimated memberships $\hat{G}_1,\dots,\hat{G}_n$ are correct in the $88\%$ of the cases. However, it should be emphasized that in both cases the correct number of cluster is automatically identified, without the need of a post-processing step. This corroborates the usage of each $H_\ell$ as an upper bound, implying that one should not be worried to overfit the data when selecting large $H$, as long as the $c_1,\dots,c_L$ parameters are well calibrated.

\section{E-commerce application}\label{sec:6}
\subsection{Prior specifications}

Recall that in our motivating application we aim at grouping flight routes according to the searches on the website of the company. From the original dataset at our disposal---concerning only Italian airports---we retained the flight routes having the highest number of searches within the period under consideration. As a result, the final dataset comprises $n = 214$ different flight routes accounting for the $94\%$ of the total counts. Each $y_i(t)$ is observed over a weekly time grid ranging from the 1st March 2017 $(t=1)$ to the 14th March 2018 $(t=55)$, so that each time grid equals $\bm{t}_i = (1,\dots,55)^\intercal$, for $i=1,\dots,n$. Hence, the dataset can be represented as a $214 \times 55$ matrix having $11770$ entries. 

\begin{figure}[tp]
\centering
\includegraphics[width=\textwidth]{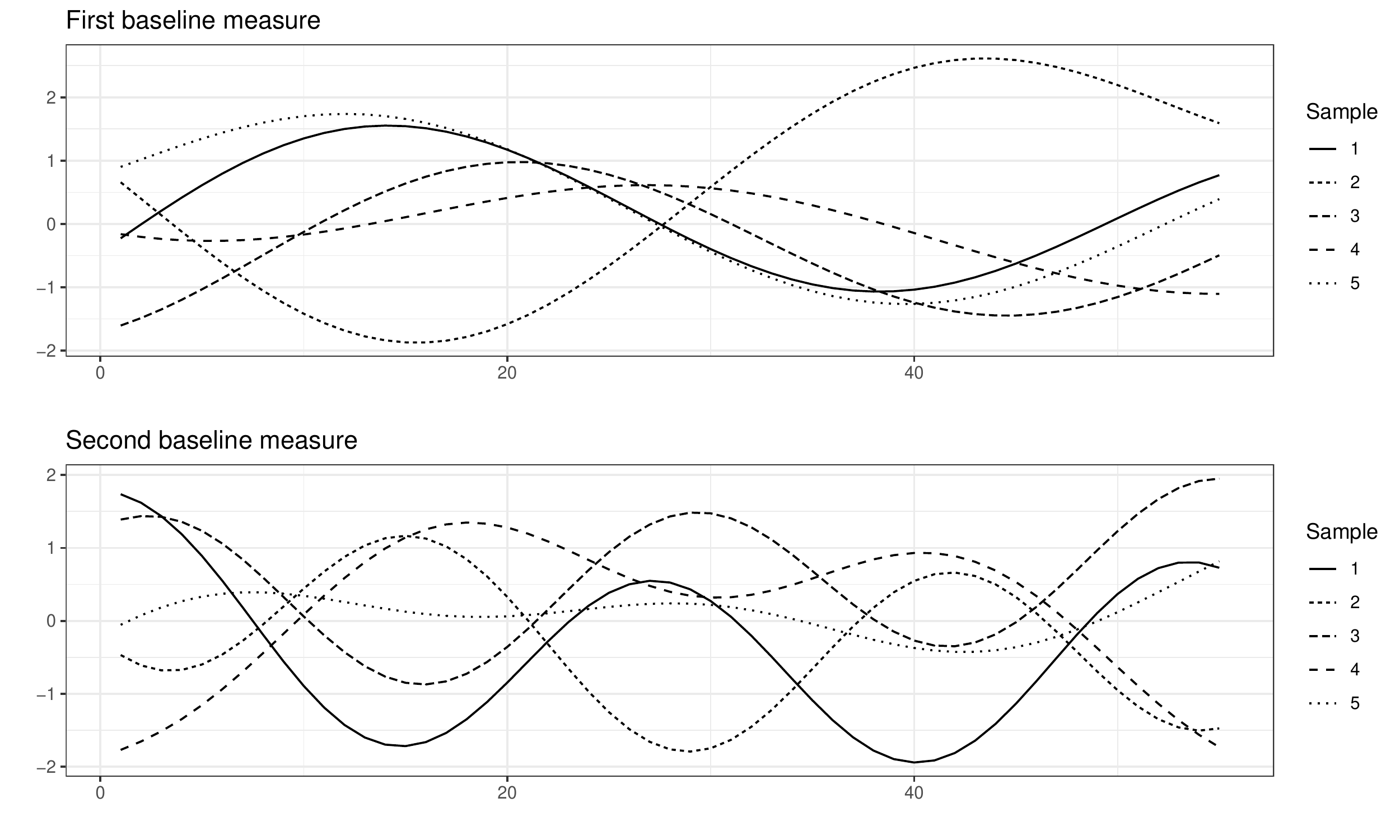}
\caption{Prior samples for the $L=2$ baseline probability measures $P_1$ (top graph) and $P_2$ (bottom graph) according to equations~\eqref{P1}-\eqref{P2}. \label{fig:2}}
\end{figure}

We set the number of functional classes $L = 2$ and we select $P_1$ and $P_2$ so that they have interpretable but yet sufficiently flexible forms. The number of basis functions for the both the functional classes is $M_1 = M_2 = 6$. The first functional class ($\ell=1$) captures yearly cyclical patterns and characterizes the routes having e.g. a peak of web-searches during either the summer or the winter. This is the case for example of the \texttt{MIL-AHO} route---from Milan to Alghero, a small city in Sardinia---as apparent from Figure~\ref{fig:1}. We increase the flexibility of this functional class by including also a semi-parametric component, thus allowing moderate deviations from this cyclical behavior. Specifically, we specialize the basis functions $\mathcal{B}_{m\ell}(t)$ in~\eqref{linear_beta} as follows
\begin{equation}\label{P1}
\theta_{1h}(t) = \sum_{m=1}^4 \beta_{m1h}\mathcal{S}_m(t) +  \beta_{51h}\cos{\left(2\pi \frac{7}{365} t\right)} + \beta_{61h}\sin{\left(2\pi \frac{7}{365} t\right)}, 
\end{equation}
where $\mathcal{S}_1(t),\dots,\mathcal{S}_4(t)$ are deterministic cubic spline basis functions. The second functional class ($\ell=2$) has a mathematical formulation similar to~\eqref{P1}, but with an important practical distinction. In particular, it characterizes functions having two peaks per year, which amounts to let
\begin{equation}\label{P2}
\theta_{2h}(t) = \sum_{m=1}^4 \beta_{m2h}\mathcal{S}_m(t) +\beta_{52h}\cos{\left(2\pi \frac{14}{365} t\right)} + \beta_{62h}\sin{\left(2\pi \frac{14}{365} t\right)}. 
\end{equation}
The \texttt{MIL-NAP} route---from Milan to Naples, depicted in Figure~\ref{fig:1}---is presumably a member of this functional class. As for the prior distributions $\bm{\beta}_{\ell h} \sim \mathcal{N}_{M_\ell}(\bm{\mu}_{\ell},\bm{\Sigma}_{\ell})$, we set the prior means $\bm{\mu}_1 = \bm{\mu}_2 = \bm{0}$ and the covariance matrices $\bm{\Sigma}_1 = \bm{\Sigma}_2$ to be equal and diagonal, having entries $\text{diag}(\bm{\Sigma}_1) = \text{diag}(\bm{\Sigma}_2)= (1,\dots,1)$, which were chosen to induce a fairly uninformative prior, considered that the data were standardized. Few simulated draws from the prior baselines $P_1$ and $P_2$ are shown in Figure~\ref{fig:2}, which confirms that these two functional classes are both sufficiently flexible but distinct.

To induce a priori a moderate amount of clusters we select $c_1 = c_2 = 1$, whereas we specify a uniform prior for functional class probabilities $\bm{\Pi} = (\Pi_1,\Pi_2)$ by letting $\alpha_1 = \alpha_2 = 1$. The latter choice corresponds to the a priori indifference between the two functional classes. Moreover, by virtue of Corollary~\ref{corol1}, it also implies that for $H_\ell$ large enough the \textsc{e-fdmp} is approximately a \textsc{fdp} with baseline measure $\frac{1}{2}(P_1 + P_2)$. Finally, we let $a_\sigma = b_\sigma = 1$ for the residual precision $\sigma^{-2}$, a fairly uninformative setting.

\subsection{Selection of the upper bounds}

The theoretical findings of Section~\ref{sec:3} as well as the simulation study of Section~\ref{sec:5} seem to suggest that each $H_\ell$ should be taken as large possible, being limited only by computational constraints. Indeed, the redundant clusters would be automatically deleted by the shrinkage prior in equation~\eqref{within_class}. Taken to the extreme (i.e.  as each $H_\ell \rightarrow \infty$), this argument would lead to a proper Bayesian nonparametric prior; see Section~\ref{sec:3}. Although such an approach is theoretically sounding, its direct application might be troublesome on certain statistical problems. Indeed, real data are far more heterogeneous than those typically considered in simulations, meaning that the ``true'' number of clusters could be large with respect to the sample size. This effect is particularly marked within the context of functional clustering, because even small local oscillations lead to mathematically distinct functions.  Hence, flexible priors with very large upper bounds---as well as infinite dimensional nonparametric priors---might constitute a better fit for the data, at the price of more complex cluster solutions. The strength of the \textsc{e-fdmp} formulation---especially in comparison with nonparametric priors---is in that one can balance the flexibility and the complexity of the model by tuning the bounds $H_\ell$.

On the basis of the above discussion, we let $H = \sum_{\ell=1}^L H_\ell$ be the largest value for which the resulting clustering solution is still useful in practice. Such a value is evidently quite subjective and it depends on the specific statistical problem. In our e-commerce application---in consultation with the stakeholders of the company---we let the upper bounds $H_1 = 20$ and $H_2 = 5$. Indeed, the second baseline measure is more prone to capture specificities of the functional observations compared to the first one, and this might lead to highly similar clusters. As discussed in the next section, such an effect is present even under the tight choice $H_2 = 5$. Note that the values $H_\ell$ still preserve their interpretation of upper bounds for the within-class number of clusters: if less than $H_\ell$ clusters are needed, then the redundant mixture components will be neglected.

\subsection{Flight routes segmentation}

We run the \textsc{cavi} Algorithm~\ref{algo1} multiple times, starting from different initialization points to mitigate the issue of local maxima. Such a procedure required only few minutes of computations on a standard laptop. From the ouput of the \textsc{cavi} algorithm, we estimate the group memberships $\hat{G}_1,\dots,\hat{G}_n$ as discussed in Section~\ref{sec:4}. In Table~\ref{tab:clusterapp} the frequencies of the resulting clusters are reported. Note that only $14$ clusters are obtained out of $H = 25$ and furthermore some of them are composed only by few functional observations. Moreover, all the $H_2 = 5$ groups of the second functional class are occupied, which suggests that by selecting a larger upper bound one would probably get more clusters. However, this would be of little practical interest because---as evidenced in Figure~\ref{fig:cluster}---these $5$ groups are already highly similar. This is an important practical advantage of the \textsc{e-fdmp} with respect to nonparametric priors, namely the ability of bounding the model complexity by avoiding the exploration of complex and less relevant partition structures. 

\begin{table}[t]
\begin{subtable}[h]{\textwidth}
\centering
\begin{tabular}{lrrrrrrrrr}
  \toprule
Within-class label & 2 & 3 & 5 & 6 & 10 & 14 & 16 & 17 & 20 \\ 
  \midrule
Frequency & 8 & 7 & 1 & 2 & 40 & 1 & 4 & 13 & 41 \\
Volume $(\times 10^5)$ & 4.49 & 2.54 & 0.51 & 0.78 & 51.45 & 0.44 & 26.61 & 15.46 & 33.43 \\
   \bottomrule
\end{tabular}
   \caption{First functional class ($\ell=1$). \label{tab:clusterapp1}}
\end{subtable}
\hfill
\begin{subtable}[h]{\textwidth}
\centering
\begin{tabular}{lrrrrr}
  \toprule
Within-class label & 1 & 2 & 3 & 4 & 5 \\ 
  \midrule
Frequency    & 27 & 9 & 28 & 21 & 12 \\ 
Volume $(\times 10^5)$ & 35.24 & 8.27 & 23.93 & 26.96 & 16.16 \\  
   \bottomrule
\end{tabular}
  \caption{Second functional class ($\ell=2$). \label{tab:clusterapp2}}
  \end{subtable}
    \caption{For both the functional classes $\ell=1$ and $\ell=2$ the frequencies of the estimated clusters, as well as the traffic volumes associated to these groups, are reported. The traffic volumes represent the summation of the within-cluster number of web-searches over the period of consideration. The cluster labels having zero frequencies are omitted. \label{tab:clusterapp}}
\end{table}

\begin{figure}[tp]
\centering
\includegraphics[width=\textwidth]{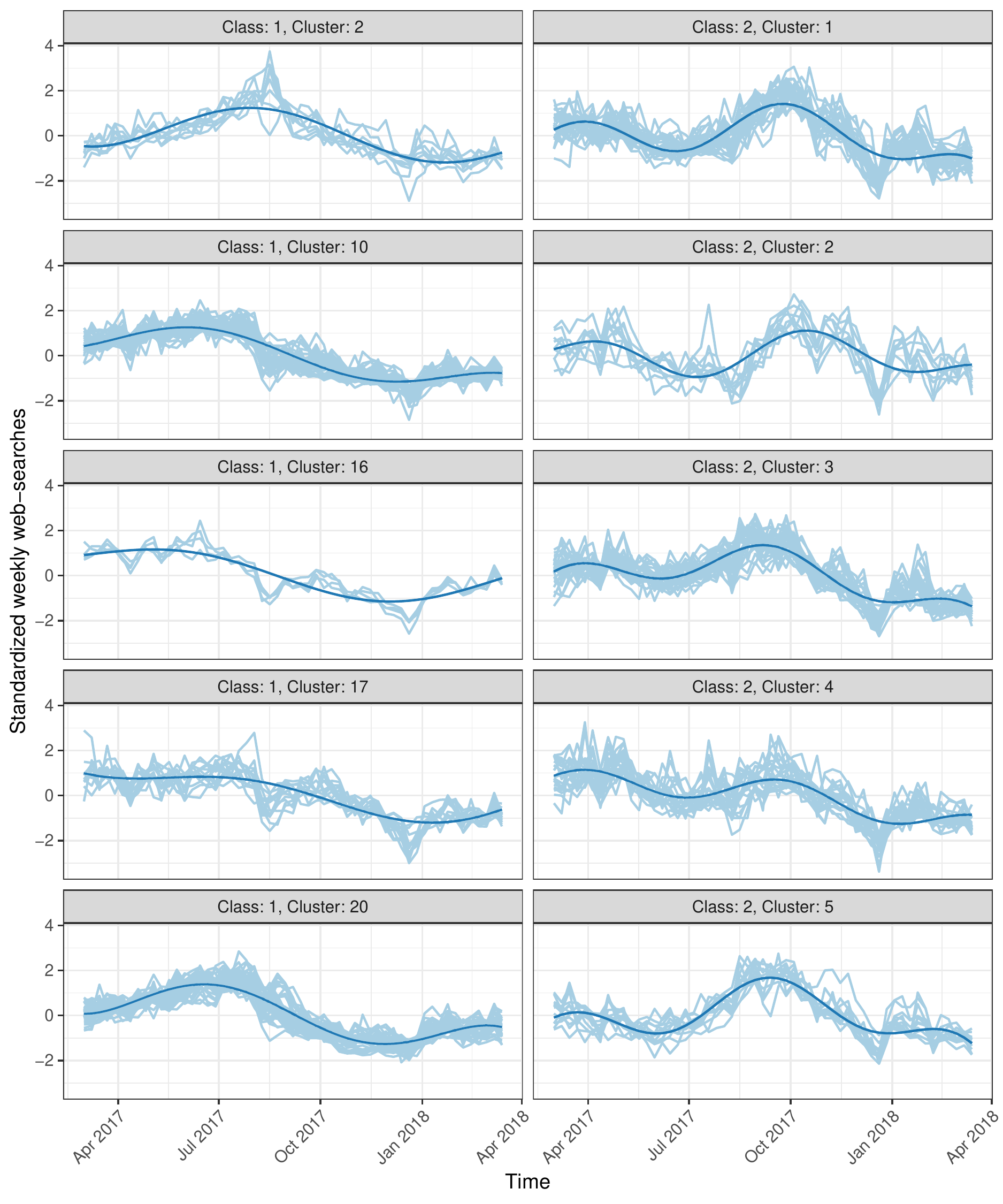}
\caption{The standardized functional observations $y_i(t)$ of the $10$ most relevant clusters (according to the volumes of Table~\ref{tab:clusterapp}) are depicted. The solid dark lines represent the associated cluster-specific estimated trajectories $\hat{\theta}_{\ell h}(t)$. \label{fig:cluster}}
\end{figure}

Together with the cluster frequencies, we report in Table~\ref{tab:clusterapp} also the traffic volumes associated to these groups, namely the within-cluster summation of the number of web-searches. Such a metric is far more important than the cluster frequencies: for example, cluster $16$ of class $1$---which has only $4$ observations and a sensible traffic volume---is much more relevant from a business perspective than cluster $3$ of class $1$. Unsurprisingly, cluster $16$ of class $1$ identifies flights from the cities Milan and Bologna to Palermo and Catania, whose airports are among the biggest in Italy. 

In Figure~\ref{fig:cluster} we depict the raw standardized observations $y_i(t)$ of the $10$ most relevant clusters---i.e. those having the highest traffic volumes---overlaid with the corresponding estimated curves $\hat{\theta}_{\ell h}(t)$. A direct graphical inspection confirms that the baseline specifications of equations~\eqref{P1}-\eqref{P2} are indeed flexible enough to capture the main tendencies of the data. Moreover, the differences between the two functional classes are evident also a posteriori: indeed, the clusters of the first column in Figure~\ref{fig:cluster} are characterized by single peaked functions, while the other groups display two-peaked functions.  

As previously mentioned, the clusters of the second functional class are mathematically different but quite similar, since all the corresponding functions have a first peak around April and a second one between September and October. Between functional classes, and within the first functional class, however, there is much more heterogeneity.  For instance, the functions belonging to cluster $2$ of class $1$ have a single peak in August, while those belonging to clusters $10$ and $20$ of class $1$ have a single peak between June and July. Moreover, functions of cluster $17$, class $1$, are quite stationary at the beginning and then they drop around August.

\begin{table}[t]
\begin{subtable}[h]{\textwidth}
\centering
\begin{tabular}{llrrr}
\toprule
& & \multicolumn{3}{c}{\emph{\textbf{Arrival}}}\\
& & North & Center & South \& Islands \\ 
\midrule
&  North &   0 &   2 &  49 \\ 
\emph{\textbf{Departure}}&  Center &   0 &   0 &  24 \\ 
&  South \& Islands &   6 &   3 &  12 \\ 
\bottomrule
\end{tabular}
   \caption{Macro cluster A. Labels $\{10,20\}$ of the first functional class $(\ell=1)$. \label{tab:transition1}}
\end{subtable}
\hfill
\begin{subtable}[h]{\textwidth}
\centering
\begin{tabular}{llrrr}
\toprule
& & \multicolumn{3}{c}{\emph{\textbf{Arrival}}}\\
& & North & Center & South \& Islands \\ 
\midrule
&North &   0 &   7 &   6 \\ 
\emph{\textbf{Departure}}&    Center &  10 &   0 &   0 \\ 
&  South \& Islands &  47 &  21 &   7 \\ 
\bottomrule
\end{tabular}
  \caption{Macro cluster B. Labels $\{1,\dots,5\}$ of the second functional class $(\ell=2)$. \label{tab:transition2}}
  \end{subtable}
    \caption{Contingency tables for the regions associated to the departure and arrival airports, for the flight routes belonging to macro clusters A and B. 
    \label{tab:transition}}
\end{table}

We now investigate in more detail the features of clusters $10$ and $20$ of the first functional class, termed henceforth macro cluster A, as well as those of the second functional class, which we will call macro cluster B. Indeed, these macro clusters are fairly homogeneous and they are also characterized by the highest traffic volumes. Recall that the airports of our dataset are located in Italy, which can be conveniently divided in three areas (North, Center and South \& Islands), following standard administrative divisions. Arrival and departure airports of the flight routes belong to one of these areas. Remarkably, both the macro clusters A and B can be well described in terms of these administrative borders, as it is apparent from Table~\ref{tab:transition}. In particular, the vast majority of flight routes belonging to macro cluster A arrive to an airport located in the South \& Island region. Conversely, in the macro cluster B most of the flight routes depart from the South \& Islands area and are directed to the North and to the Center regions. These findings further corroborate the quality of the obtained cluster solution and they provide useful intuitions about the role of each cluster. Indeed, these qualitative descriptions might help marketing specialists in designing effective cluster-specific policies. 

\section*{Acknowledgements}The author is grateful to Gianluca Barbierato and Bruno Scarpa for their helpful comments on a first version of this manuscript.

\bibliographystyle{chicago}
\bibliography{biblio}

\end{document}